\documentclass{sig-alternate-2013}

\newfont{\mycrnotice}{ptmr8t at 7pt}
\newfont{\myconfname}{ptmri8t at 7pt}
\let\crnotice\mycrnotice%
\let\confname\myconfname%

\permission{Permission to make digital or hard copies of all or part of this work for personal or classroom use is granted without fee provided that copies are not made or distributed for profit or commercial advantage and that copies bear this notice and the full citation on the first page. Copyrights for components of this work owned by others than ACM must be honored. Abstracting with credit is permitted. To copy otherwise, or republish, to post on servers or to redistribute to lists, requires prior specific permission and/or a fee. Request permissions from permissions@acm.org.}
\conferenceinfo{PODC'15,}{July 21--23, 2015, Donostia-San Sebasti\'an, Spain.}
\copyrightetc{Copyright \copyright~2015 ACM \the\acmcopyr}
\crdata{978-1-4503-3617-8 /15/07\ ...\$15.00.\\
DOI:http://dx.doi.org/10.1145/2767386.2767446 }

\usepackage{thmtools, thm-restate}
\usepackage{stmaryrd}
\usepackage[small,margin=1cm]{caption}
\usepackage{microtype}
\usepackage{hyperref}
\usepackage{todonotes}
\usepackage{comment}

\newtheorem{theorem}{Theorem}
\newtheorem{lemma}[theorem]{Lemma}

\begin{document}
%%%%%%%%%%%%%%%%%%%%%%%%%%%%%%%%%%%%%%%%%%

\title{Brief Announcement : Average Complexity for the LOCAL Model}

\author{Laurent Feuilloley\thanks{\small \sl Part of this work was done while the author was visiting LIAFA at university Paris Diderot, with additional support from ANR project DISPLEXITY.}\\
{\small\'Ecole Normale Sup\'erieure de Cachan,}\\  
{\small France}\\
}

\maketitle

\abstract{A standard model in network synchronised distributed computing is the LOCAL model \cite{Peleg00}. In this model, the processors work in rounds and, in the classic setting, they know the number of vertices of the network, $n$. Using $n$, they can compute the number of rounds after which they must all stop and output. It has been shown recently that for many problems, one can basically remove the assumption about the knowledge of $n$, without increasing the asymptotic running time \cite{KormanSV13}\cite{Musto11}. In this case, it is assumed that different vertices can choose their final output at different rounds, but continue to transmit messages. In both models, the measure of the running time is the number of rounds before the last node outputs. In this brief announcement, the vertices do not have the knowledge of $n$, and we consider an alternative measure: the average, over the nodes, of the number of rounds before they output.   We prove that the complexity of a problem can be exponentially smaller with the new measure, but that Linial's lower bound for colouring \cite{Linial92} still holds.}

\section{Introduction}
In the LOCAL model \cite{Peleg00}, the processors are located at the nodes of the network, have distinct identifiers, and work in rounds. At each round, each processor sends messages to its direct neighbours, receives messages from them, and computes its new state. 
In some variants of this model, e.g., when $n$ is not known \cite{KormanSV13,Musto11}, every node can choose its output at an arbitrary round, yet it must continue to transmit the messages it receives. The classic measure of the running time is the number of rounds before all the nodes have output. We consider an alternative measure which is the average, over the nodes, of the running time before they output. 

An equivalent way to describe the LOCAL model is to consider that every node gathers all the information in a ball around itself and output a function of this ball. For example a node can increment the radius of the ball it sees, until it has enough information to output. This second vision of the LOCAL model is more convenient for this paper. Therefore, in the following, we mostly consider radiuses, and not rounds. When the algorithm, the graph, and the identifiers are set, the ``radius of the node $v$'' refers to the radius at which the algorithm chooses to output, and we denote it by $r(v)$. From this point of view, the classic measure of the running time is the maximum of the radiuses of these balls, and the alternative measure is the average of these radiuses. More precisely for a given size $n$, the running time is usually: 
$$\max_{G:|G|=n} \left( \max_{v\in G}r(v)\right),$$ 
and, in this paper, we consider:  
$$\max_{G:|G|=n}\left( \sum_{v\in G}r(v)/n \right).$$
 
In this brief announcement, the computation is always deterministic, and the average is always over the nodes. In particular, remark that we consider the worst case for the distribution of the identifiers. 

On one hand, we show that there is a natural problem for which there exists an algorithm with average running time exponentially smaller than the worst case complexity. In this problem, called \emph{largest ID}, every vertex of a cycle must decide if it has the largest ID or not, which is a classic way to elect a leader. On the other hand, we show that Linial's lower bound on colouring \cite{Linial92} holds for the average node measure, i.e. the vertices need an average radius of $\Omega(\log^*n)$ to compute a valid 3-colouring. Note that this lower bound matches the upper bound as it is possible to 3-colour the $n$-node ring in $O(\log^*n)$ rounds even without the knowledge of $n$ \cite{KormanSV13}\cite{Musto11}.

The goal of this work is to continue the study of locality by proposing a new measure of the running time. This measure is suited to algorithms in which some vertices may stop very early. An example of application is in the context of dynamic networks. The average time to update the labels of the graph after a change at a random node, can be estimated using the average measure. Also in the context of parallel computations that simulate distributed computations, we can take advantage of the fact that a job is finished earlier to process an other job, and then the average running time is the relevant measure. 

\section{Algorithm for the largest ID in a cycle}

In this section, we show that the \emph{largest ID} problem on a cycle has a linear worst case complexity, and that there exists an algorithm with logarithmic average radius. In this problem, each vertex must output \emph{Yes} if it has the largest identifier in the graph, and \emph{No} otherwise. 
The worst case complexity of the problem is linear because, for any algorithm, the vertex with the maximum ID needs to see all the cycle. The following straightforward algorithm gives a linear upper bound: each node increases its radius until it discovers an ID that is larger than its own ID, or until it has seen all the cycle. We show that the average running time of this algorithm is logarithmic in $n$.

The vertex with the maximum identifier needs $n/2$ rounds to discover that it has the largest ID. We can consider that the other vertices are in a path, and that reaching an endpoint is sufficient to stop and output \emph{No}. 
The radius needed by these other vertices is the minimum distance to a vertex with larger ID or an endpoint. Note that the vertex with largest ID of the segment must reach an endpoint. We can then subdivide the vertices of the segment into three parts: the one with the largest ID, the ones in the segment on the left, and the ones on the right. Let $a(p)$ be the maximum (over the permutations of the identifiers) sum of radiuses in a segment with $p$ vertices. Then the following recurrence relation follows from the decomposition into three pieces, and from the symmetry:
$$a(p) = \max_{1 \leq k \leq \lceil p/2\rceil}\left\{ k + a(k-1) + a(p-k) \right\}.$$ 

This sequence $a(n)$, with initial values $a(0)=0$, and $a(1)=1$, is known to be in $\theta(n\ln(n))$ (see for example the sequence A000788 of the OEIS \cite{oeisA000788}). Thereafter, the average radius is logarithmic in $n$, which is exponentially smaller than the worst case complexity.

\section{Lower bound for colouring}
The classic algorithm for distributed 3-colouring a ring is the Cole-Vishkin algorithm \cite{ColeV86}, that uses $O(\log^*n)$ rounds for every vertex. This is basically optimal for the classic measure, as Linial proved that this task requires $\Omega(\log^*n)$ rounds\cite{Linial92}. One could try to improve the average complexity, using less rounds for some vertices. In this section, we show that this is useless: Linial's lower bound also holds for this new measure.

% Theorem
\begin{theorem}\label{thm:col}
The average complexity of 3-colouring a ring with $n$ nodes is $\Omega(\log^*n)$.
\end{theorem} 

We say that an algorithm $A$ for 3-colouring is minimal, if there is no other algorithm $A'$ that behaves strictly better. More precisely, $A$ is minimal if there does not exist $A'$ that uses at most the same radius as $A$ for every neighbourhood, and a strictly smaller radius for at least one neighbourhood. Without loss of generality we prove the theorem only for minimal algorithms. We first prove two lemmas about the regularity of the distribution of the radiuses. For technical reasons, we begin with 4-colouring and then come back to 3-colouring.

\vspace{0.5cm}

%lemma A
\begin{lemma}\label{lem:threshold}
In a graph $G$ with identifiers, given two arbitrary vertices $x$ and $y$, separated by $k$ vertices, if an algorithm $A$ is minimal for 4-colouring, then the radiuses of the vertices between $x$ and $y$ are at most $\max\{r(x),r(y)\}+k$.
\end{lemma}
 
\begin{proof} For the sake of contradiction, suppose that there exists $G$, $x$, $y$, $k$ and $A$ as in the lemma, and some vertices between $x$ and $y$ with radiuses strictly larger than the threshold, $\max\{r(x),r(y)\}+k$. Let $N$ be the slice of identifiers that contains $x$, $y$, the vertices between them, and the views of $x$ and $y$. We show how to transform $A$ into a strictly smaller algorithm $A'$, by decreasing the radiuses of some vertices in the neighbourhood $N$. For every vertex $v$, if when running $A$, $v$ discovers that it is not in the neighborhood $N$ between $x$ and $y$, or if it stops before reaching the threshold, then it does the same with $A'$. Otherwise, $v$ knows that it is in this particular neighbourhood, stops at radius $\max\{r(x),r(y)\}+k$, and outputs obeying two simple rules that depend only on the two direct neighbours. Assume without loss of generality that $ID(x)>ID(y)$, and let $d$, be the distance between $v$ and $x$.
First, if a neighbour has stopped strictly before the threshold, choose a different colour. Second, if a neighbour has not stopped before the threshold, then if $d$ is even, $v$ outputs a colour in $\{1,2\}$, else $v$ outputs a colour in $\{3,4\}$. 
One can check that this new algorithm $A'$ produces a valid 4-colouring for every graph and identifiers. 
Moreover for every neighbourhood the radius given by $A'$ is at most the one given by $A$, and, for at least one neighbourhood, strictly smaller. This contradicts the fact that $A$ is minimal.
\end{proof}
 
We use lemma \ref{lem:threshold} to prove the following more practical result.  
 
%lemma B
%
 \begin{lemma}\label{lem:neighbourhood}
If a vertex $v$ uses radius $r$, in a minimal algorithm $A$ for 3-colouring, then the average of the radiuses of the vertices at distance at most $r/2$ from $v$, is $\Omega(r)$.
\end{lemma} 

\begin{proof} First note that it is sufficient to prove the result for 4-colouring, as it implies the result for 3-colouring. 
Let $d$ in $\{1,2,..., \lfloor{} r/2 \rfloor{} \}$, and let $u_d$ and $w_d$ be the two vertices at distance $d$ from $v$ on the cycle. 
Lemma \ref{lem:threshold} implies 
$$r \leq \max\{r(u_d),r(w_d)\} + 2d -1,$$
then 
$$r- 2d +1 \leq r(u_d)+ r(w_d).$$ 
Then by summing over $d$:
$$ \sum_{d=1}^{\lfloor{} r/2 \rfloor{}} (r -2d +1) 
\leq  \sum_{d=1}^{\lfloor r/2 \rfloor{}} r(u_d) +r(w_d).$$
If we add $r$ on both side of the equation, the right-hand term is the sum of the radiuses of the vertices at distance at most $r/2$ from $v$, and the left-hand term is quadratic in $r$. The lemma follows. 
\end{proof}

The end of the proof of theorem \ref{thm:col} uses the following corollary of Linial's lower bound as a black box. For every algorithm that 3-colours a cycle of length larger than $n/2$, there exists a permutation of the node identifiers such that at least one vertex needs a radius of $\frac{1}{2}.\log^*(n/2)$. Given a minimal algorithm $A$, we show that we can build a permutation $\pi$ of the identifiers that leads to an average running time in $\Omega(\log^*n)$. 

First, consider an $n$-cycle and a permutation of the identifiers, such that one vertex has radius at least $\frac{1}{2}.\log^*(n/2)$. We take the slice of identifiers that are in the ball of radius $\frac{1}{2}.\log^*(n/2)$ around this vertex, and put it at the beginning of $\pi$. Then, consider the rest of the identifiers, and repeat the operation: remove from the original set of identifier a $\frac{1}{2}.\log^*(n/2)$-ball around a vertex with large radius, and concatenate it to $\pi$. We do it until there are less than $n/2$ vertices remaining in the original set. Finally, we put the rest of the identifiers at the end of $\pi$ in an arbitrary order. 

The vertex at the centre of each slice has exactly the same $\frac{1}{2}\log^*(n/2)$-neighbourhood as when it was removed from the original set of identifiers, hence its radius in $\pi$, with the same algorithm $A$, is at least $\frac{1}{2}\log^*(n/2)$. Thanks to lemma \ref{lem:neighbourhood}, the average radius in each slice is $\Omega(\log^*(n/2))$, that is $\Omega(\log^*(n))$. Then the average radius in $\pi$ is $\Omega(\log^*(n))$. This concludes the proof of the theorem.

\section{Conclusion and further work}

This paper presents a new measure of the locality. It is close to the classic measure for some problems, and very different for others. It would be interesting to characterise the problems of the first and second types. Also, we only consider the cycle topology, and results for more general graphs are missing. Last, as the average we consider is over the nodes, but it would also be interesting to begin to study the expectancy of the running time on graphs where the permutation of the identifiers is taken uniformly at random, for both the classic and the new measure.

\section*{Acknowledgements}
I would like to thank Pierre Fraigniaud, Juho Hirvonen, Tuomo Lempi\"ainen and Jukka Suomela for helpful discussions.

\end{document}